\documentclass[pra,aps,twocolumn,superscriptaddress]{revtex4-1}
\usepackage[dvips]{graphicx}
\usepackage{amsmath,amssymb,amsthm}
\usepackage[matrix,frame,arrow]{xy}
\usepackage{amsmath}
\newcommand{\bra}[1]{\langle{#1}\vert}
\newcommand{\ket}[1]{\vert{#1}\rangle}
\newcommand{\qw}[1][-1]{\ar @{-} [0,#1]}

\newcommand{\cw}[1][-1]{\ar @{=} [0,#1]}

\newcommand{\gate}[1]{*{\xy *+<.6em>{#1};p\save+LU;+RU **\dir{-}\restore\save+RU;+RD **\dir{-}\restore\save+RD;+LD **\dir{-}\restore\POS+LD;+LU **\dir{-}\endxy} \qw}

\newcommand{\cgate}[1]{*{\xy *+<.6em>{#1};p\save+LU;+RU **\dir{-}\restore\save+RU;+RD **\dir{-}\restore\save+RD;+LD **\dir{-}\restore\POS+LD;+LU **\dir{-}\endxy} \cw}

\newcommand{\measureD}[1]{*{\xy*+=+<.5em>{\vphantom{\rule{0em}{.1em}#1}}*\cir{r_l};p\save*!R{#1} \restore\save+UC;+UC-<.5em,0em>*!R{\hphantom{#1}}+L **\dir{-} \restore\save+DC;+DC-<.5em,0em>*!R{\hphantom{#1}}+L **\dir{-} \restore\POS+UC-<.5em,0em>*!R{\hphantom{#1}}+L;+DC-<.5em,0em>*!R{\hphantom{#1}}+L **\dir{-} \endxy} \qw}

\newcommand{\multimeasureD}[2]{*+<1em,.9em>{\hphantom{#2}}\save[0,0].[#1,0];p\save !C *{#2},p+LU+<0em,0em>;+RU+<-.8em,0em> **\dir{-}\restore\save +LD;+LU **\dir{-}\restore\save +LD;+RD-<.8em,0em> **\dir{-} \restore\save +RD+<0em,.8em>;+RU-<0em,.8em> **\dir{-} \restore \POS !UR*!UR{\cir<.9em>{r_d}};!DR*!DR{\cir<.9em>{d_l}}\restore \qw}

\newcommand{\multigate}[2]{*+<1em,.9em>{\hphantom{#2}} \qw \POS[0,0].[#1,0];p !C *{#2},p \save+LU;+RU **\dir{-}\restore\save+RU;+RD **\dir{-}\restore\save+RD;+LD **\dir{-}\restore\save+LD;+LU **\dir{-}\restore}
\newcommand{\cmultigate}[2]{*+<1em,.9em>{\hphantom{#2}} \cw \POS[0,0].[#1,0];p !C *{#2},p \save+LU;+RU **\dir{-}\restore\save+RU;+RD **\dir{-}\restore\save+RD;+LD **\dir{-}\restore\save+LD;+LU **\dir{-}\restore}
\newcommand{\ghost}[1]{*+<1em,.9em>{\hphantom{#1}} \qw}
\newcommand{\cghost}[1]{*+<1em,.9em>{\hphantom{#1}} \cw}

\newcommand{\gategroup}[6]{\POS"#1,#2"."#3,#2"."#1,#4"."#3,#4"!C*+<#5>\frm{#6}}
\newcommand{\rstick}[1]{*!L!<-.5em,0em>=<0em>{#1}}
\newcommand{\lstick}[1]{*!R!<.5em,0em>=<0em>{#1}}
\newcommand{\ustick}[1]{*!D!<0em,-.5em>=<0em>{#1}}

\newcommand{\Qcircuit}[1][0em]{\xymatrix @*=<#1>}

\newcommand{\pureghost}[1]{*+<1em,.9em>{\hphantom{#1}}}
\newcommand{\multiprepareC}[2]{*+<1em,.9em>{\hphantom{#2}}\save[0,0].[#1,0];p\save !C
  *{#2},p+RU+<0em,0em>;+LU+<+.8em,0em> **\dir{-}\restore\save +RD;+RU **\dir{-}\restore\save
  +RD;+LD+<.8em,0em> **\dir{-} \restore\save +LD+<0em,.8em>;+LU-<0em,.8em> **\dir{-} \restore \POS
  !UL*!UL{\cir<.9em>{u_r}};!DL*!DL{\cir<.9em>{l_u}}\restore}
\newcommand{\prepareC}[1]{*{\xy*+=+<.5em>{\vphantom{#1\rule{0em}{.1em}}}*\cir{l^r};p\save*!L{#1} \restore\save+UC;+UC+<.5em,0em>*!L{\hphantom{#1}}+R **\dir{-} \restore\save+DC;+DC+<.5em,0em>*!L{\hphantom{#1}}+R **\dir{-} \restore\POS+UC+<.5em,0em>*!L{\hphantom{#1}}+R;+DC+<.5em,0em>*!L{\hphantom{#1}}+R **\dir{-} \endxy}}

\newcommand{\Tr}{\operatorname{Tr}}
\newcommand{\lin}[1]{\mathcal{L}(#1)}
\newcommand{\hilb}[1]{\mathcal{#1}}
\newcommand{\set}[1]{\mathcal{#1}}
\newcommand{\map}[1]{\mathcal{#1}}

\newcommand{\erf}{\operatorname{erf}}
\newcommand{\arcsinh}{\operatorname{arcsinh}}

\newtheorem{dfn}{Definition}
\newtheorem{lmm}{Lemma}
\newtheorem{thm}{Proposition}

\def \qcwidth {0.7em}
\def \qcheight {0.6em}

\begin{document}

\title{Gaussian Quantum Reading beyond the Standard Quantum Limit}

\author{Michele \surname{Dall'Arno}}

\affiliation{Graduate School of Information Science, Nagoya
  University, Chikusa-ku, Nagoya, 464-8601, Japan}

\date{\today}

\begin{abstract}
  Quantum reading aims at retrieving classical information stored in
  an optical memory with low energy and high accuracy by exploiting
  the inherently quantum properties of light. We provide an optimal
  Gaussian strategy for quantum reading with phase-shift keying
  encoding that makes use of squeezed coherent light and homodyne
  detectors to largely outperform the Standard Quantum Limit, even in
  the presence of loss. This strategy, being feasible with current
  quantum optical technology, represents a viable prototype for a
  highly efficient and reliable quantum-enhanced optical reader.
\end{abstract}

\maketitle

\section{Introduction}

Digital optical memories store classical information in the optical
properties of a media. They are ubiquitous in modern applications of
classical information processing, such as CDs or DVDs. To retrieve
information, optical readers illuminate each memory cell with a light
probe and measure the outgoing signal. Clearly a tradeoff exists
between the energy of the source and the information it is able to
extract. On the one hand, the need for miniaturized and embedded
technology push toward the use of low energetic light sources. On the
other hand, reliability is an unavoidable requirement for modern
information processing technology.

In the general framework of quantum theory, inherently quantum
properties of light can be exploited to maximize the amount of
retrieved information per fixed energy of the light probe. This
problem, usually referred to as quantum reading, was recently
introduced by Pirandola~\cite{Pir11} and immediately triggered a
noticeable
interest~\cite{BDD11,Nai11,PLGMB11,Hir11,DBDMJD12,SLMBP12,PUR12,NYGSP12}.
However, most of the experimental as well as commercial
implementations of optical readers make use of a suboptimal strategy
exploiting coherent light probes generated by a laser beam. The
maximal amount of information per fixed energy that can be retrieved
using a coherent probe defines the Standard Quantum
Limit~\cite{GLM04,Oza89,Oza88} in quantum reading. Since in general
Gaussian states and measurements are experimentally feasible with
current optical technology, the question naturally arises: can
Gaussian quantum reading outperform the Standard Quantum Limit, thus
providing highly efficient while practically feasible quantum reading
techniques?

We need to be more specific before answering this question. In
applications, the most common ways to encode classical information are
in the phase or in the amplitude of the signal. The two methods are
usually referred to as phase-shift keying (PSK) and amplitude-shift
keying (ASK), respectively. Due to its simplicity, PSK encoding has
been widely adopted in several ISO and IEEE standards, such as
wireless LAN (wifi), several credit cards, Bluetooth, and satellite
communications. Despite its widespread use, the problem of quantum
reading was mainly addressed in the context of APK
encoding~\cite{Pir11,Nai11,PLGMB11,SLMBP12,PUR12}, while only recently
the interest drew by quantum reading with PSK encoding increased.

In Ref.~\cite{Hir11} the problem of quantum reading of two signals
with phase difference $\pi$ was addressed in the lossless scenario. It
was shown that perfect quantum reading can be achieved by a particular
class of entangled coherent states~\cite{San12}. Very recently, in
Ref.~\cite{NYGSP12} the scenario was generalized to the case of $M$
signals with symmetrically distributed phases - namely, the relative
phase of signal $i$ with respect to a given seed state is $2\pi
i/M$. In particular, it was shown that squeezed coherent states
outperform coherent ones in quantum reading of two signals with phase
difference of $\pi$ in the lossless scenario and without any
constraint on the measurement applied. While of the utmost importance
these results leave the aforementioned question open: can Gaussian
quantum reading - namely, quantum reading with Gaussian probes and
Gaussian measurements~\cite{WPGCRSL12} - outperform the Standard
Quantum Limit in the general lossy scenario?

The aim of this work is to affirmatively answer this question in the
context of PSK encoding. We provide an optimal Gaussian strategy
exploiting squeezed coherent light and homodyne detection to perform
quantum reading for any value of the phase difference between the two
signals and in the presence of loss. A comparison of the optimal
Gaussian strategy with the Standard Quantum Limit shows that the
former largely outperforms the latter, even in the presence of loss
and taking into account present technological limitations in the
preparation of highly squeezed states. The proposed optimal strategy
is suitable for implementation with current quantum optical
technology, thus representing a viable prototype for a highly
efficient and reliable quantum-enhanced optical reader.

The paper is structured as follows. In Sect.~\ref{sect:qread} we
introduce the problem of quantum reading and simplify it under the
assumption that no encoding is done in the amplitude of the signal
(e.g. PSK encoding). In Sect.~\ref{sect:gaussread} we provide the
Standard Quantum Limit and an optimal Gaussian strategy for quantum
reading with PSK encoding. In Sect.~\ref{sect:comp} we compare the
optimal Gaussian strategy with the Standard Quantum Limit and
demonstrate experimental feasibility of Gaussian quantum reading
beyond the Standard Quantum Limit. Finally we summarize our results
and propose future developments in Sect.~\ref{sect:concl}.

\section{Quantum reading}
\label{sect:qread}

In this Section we formally introduce the quantum reading of optical
memories as the problem of determining the tradeoff between energy and
probability of error in the discrimination of quantum channels. Let us
first fix the notation~\cite{NC00}.

A $m$-modes {\em quantum optical setup} is represented by an Hilbert
space $\hilb{H} = \bigotimes_{i=1}^m \hilb{H}_i$, where $\hilb{H}_i$
is the Fock space representing mode $i$ and $a_i$ is the corresponding
annihilation operator. In the following we denote with
$\lin{\hilb{H}}$ the space of linear operators on $\hilb{H}$. A {\em
  quantum state} on $\hilb{H}$ is described by a density matrix $\rho
\in \lin{\hilb{H}}$, namely a positive semidefinite operator
satisfying $\Tr[\rho] = 1$, and in the following we pictorially
represent it with $\Qcircuit @C=\qcwidth @R=\qcheight
{\prepareC{\rho}&\qw}$. Up to irrelevant constants its {\em energy} is
given by $E(\rho) := \Tr[N \rho]$, where $N := \sum_{i=1}^m
a_i^\dagger a_i$ is the {\em number operator} on $\hilb{H}$. A {\em
  quantum measurement} on $\hilb{H}$ is described by a POVM $\Pi$,
namely a map associating to any state $\rho \in \lin{\hilb{H}}$ a
probability distribution $p(j|\rho)$ over a set $\{j\}$ of outcomes,
and we represent it with $\Qcircuit @C=\qcwidth @R=\qcheight {&
  \measureD{\Pi}}$. The most general transformation from states to
states is described by a {\em quantum channel} $\map{C} :
\lin{\hilb{H}} \to \lin{\hilb{H}'}$, namely a completely positive and
trace preserving map that we represent with $\Qcircuit @C=\qcwidth
@R=\qcheight {&\gate{\map{C}}&\qw}$.

The most general strategy to discriminate a channel $\map{C}_i$
randomly chosen from a set $\{\map{C}_i : \lin{\hilb{H}} \to
\lin{\hilb{H}}\}_{i=1}^d$ with probability $p_i$ consists in probing
it with a state $\rho \in \lin{\hilb{H}\otimes\hilb{K}}$, where
$\hilb{K}$ is an ancillary space, and measuring the output state with
a $d$-outcomes~\cite{note:povm} POVM $\Pi$, namely
\begin{align}\label{eq:setup1}
  \begin{aligned}
    \Qcircuit @C=\qcwidth @R=\qcheight { \multiprepareC{1}{\rho} &
      \ustick{\hilb{H}} \qw & \gate{\map{C}_i} &
      \multimeasureD{1}{\Pi}\\ \pureghost{\rho} & \ustick{\hilb{K}}
      \qw & \qw & \ghost{\Pi} }
  \end{aligned}.
\end{align}
The probability of error is given by $P_e(\rho,\Pi) := 1-\sum_{i=1}^d
p_i p(i|\map{C}_i(\rho))$. Deriving the tradeoff between
$P_e(\rho,\Pi)$ and $E(\rho)$ is the aim of quantum reading.

\begin{dfn}[Quantum reading]
  \label{def:qread}
  The optimal state $\rho^*$ and the optimal measurement $\Pi^*$ for
  quantum reading of channels $\{\map{C}_i\}$ distributed according to
  probability $\{p_i\}$ are those that minimize the error-probability
  $P_e(\rho,\Pi)$ while satisfying $E(\rho) \le E$ for given energy
  threshold $E$, namely
  \begin{align*}
    (\rho^*,\Pi^*) = \arg\min_{(\rho,\Pi)} P_e(\rho,\Pi) \textrm{
      s.t. } E(\rho) \le E.
  \end{align*}
\end{dfn}

Let us specify the problem more. An $m$-modes {\em quantum optical
  device}~\cite{Leo03} is described by a unitary operator $U :
\hilb{H} \to \hilb{H}'$ relating $d$ input optical modes with
annihilation operators $a_i$'s on Fock space $\hilb{H}_i$ to $d$
output optical modes with annihilation operators $a'_i$'s on Fock
space $\hilb{H}'_i$. An optical device $U$ is called {\em linear} and
{\em passive} if the $a_i$'s are linearly related to the $a'_i$'s,
namely $\bar{a}' = S_U \bar{a}$, where $\bar{a} = (a_1,\dots a_d)$ and
$S_U$ is the {\em scattering matrix} associated to $U$. A passive
device $U$ conserves the energy, namely for any state $\rho$ one has
$E(\rho) = E(U \rho U^\dagger)$. A simple example of quantum optical
device is the $\psi$-{\em{phase shifter}}, namely a single mode device
represented by the unitary $P^\phi = \exp(i \phi a^\dagger
a)$. Another example is the {\em beamsplitter} with transmittivity
$\eta$, namely a $2$-modes device represented by the unitary $B^\eta =
\exp[\theta(a^\dagger b - a b^\dagger)]$, where $\eta =
\cos^2\theta$. For any unitary $U$ we denote with $\map{U}(\rho) := U
\rho U^\dagger$ the corresponding {\em unitary channel}; we will
consider only linear and passive unitary channels.

Loss, affecting any experimental quantum optical implementation, can
be modeled by a {\em lossy channel} $\map{E}^\eta(\rho) :=
\Tr_1[\map{B}^\eta (\rho_0 \otimes \ket{0}\bra{0}_1)]$ with {\em
  quantum efficiency} $\eta$, namely a beamsplitter with signal and
vacuum as inputs, and one output mode of which is traced out, or
equivalently
\begin{align*}
  \begin{aligned}
    \Qcircuit @C=\qcwidth @R=\qcheight { & \ustick{\hilb{H}_0} \qw &
      \gate{\map{E}^{\eta}} & \qw }
  \end{aligned}
  =
  \begin{aligned}
     \Qcircuit @C=\qcwidth @R=\qcheight { & \ustick{\hilb{H}_0} \qw &
       \multigate{1}{\map{B}^\eta} & \qw \\ \prepareC{\ket{0}} &
       \ustick{\hilb{H}_1} \qw & \ghost{\map{B}^\eta} & \measureD{I} }
  \end{aligned}.
\end{align*}
For any bipartite operator $X_{01} \in \lin{\hilb{H}_0 \otimes
  \hilb{H}_1}$, we denote with $\Tr_1[X_{01}]$ the partial trace over
Hilbert space $\hilb{H}_1$. When loss affects $m$ optical modes we
write $\map{E}^{\bar\eta} := \bigotimes_i \map{E}^{\eta_i}$ with
$\bar\eta = (\eta_1,\dots \eta_m)$. A {\em lossy device}
$\map{U}^{\bar\eta}$ is described by the composition of a unitary
channel $\map{U}$ and a lossy channel $\map{E}^{\bar\eta}$, namely
$\map{U}^{\bar\eta} := \map{U} \circ \map{E}^{\bar\eta}$. Analogously,
a {\em lossy source} $\rho^{\bar\eta}$ is described by the preparation
of an ideal state $\rho$ followed by a lossy channel
$\map{E}^{\bar\eta}$, namely $\rho^{\bar\eta} := \map{E}^{\bar\eta}
(\rho)$, and a {\em lossy measurement} $\Pi^{\bar\eta}$ is described
by an ideal one $\Pi$ preceded by a lossy channel
$\map{E}^{\bar\eta}$, namely $\Pi^{\bar\eta} :=
\map{E}^{{\bar\eta}\vee}(\Pi)$, where $\map{C}^\vee$ represents the
application of channel $\map{C}$ in the Heisenberg picture. In
principle lossy channels can be absorbed in the definition of states
and measurements, nevertheless in the following it will be convenient
to keep the contribution of loss in evidence.

When each memory cell can be modeled as a lossy optical device, the
quantum reading strategy given by Eq.~\eqref{eq:setup1} becomes
\begin{align}\label{eq:setup2}
  \begin{aligned}
    \Qcircuit @C=\qcwidth @R=\qcheight { \multiprepareC{1}{\rho} &
      \multigate{1}{\map{E}^{\bar\alpha}} & \ustick{\hilb{H}} \qw &
      \gate{\map{E}^{\bar\beta_i}} & \gate{\map{U}_i} &
      \multigate{1}{\map{E}^{\bar\gamma}} &
      \multimeasureD{1}{\Pi}\\ \pureghost{\rho} &
      \ghost{\map{E}^{\bar\alpha}} & \ustick{\hilb{K}} \qw & \qw & \qw
      & \ghost{\map{E}^{\bar\gamma}} & \ghost{\Pi}
      \gategroup{1}{1}{2}{2}{2mm}{--} \gategroup{1}{4}{1}{5}{2mm}{--}
      \gategroup{1}{6}{2}{7}{2mm}{--}}
  \end{aligned}.
\end{align}
Here and in the following, the dashed box on the left surrounds the
lossy preparation, the one in the middle the lossy unknown device,
while the one on the right the lossy measurement. The setup in
Eq.~\eqref{eq:setup2} can be simplified taking into account the
following rule of composition of lossy channels.

\begin{lmm}[Composition]
  \label{thm:composition}
  Given an Hilbert space $\hilb{H}_0$, the composition of lossy
  channels $\map{E}^{\bar\alpha} : \lin{\hilb{H}_0} \to \lin{\hilb{H}_0}$
  and $\map{E}^{\bar\beta} : \lin{\hilb{H}_0} \to \lin{\hilb{H}_0}$ is a
  lossy channel $\map{E}^{\bar\eta} : \lin{\hilb{H}_0} \to
  \lin{\hilb{H}_0}$ with efficiency $\bar\eta = (\alpha_1\beta_1,
  \alpha_2\beta_2, \dots, \alpha_m\beta_m)$, namely
  \begin{align*}
    \map{E}^{\bar\alpha} \circ \map{E}^{\bar\beta} (\rho) =
    \map{E}^{\bar\eta} (\rho), \qquad \forall \rho \in \lin{\hilb{H}_0}.
  \end{align*}
\end{lmm}

\begin{proof}
  Since $\map{E}^{\bar\eta} = \bigotimes_i \map{E}^{\eta_i}$, it is
  sufficient to prove the statement for a single mode, namely we prove
  that $\map{E}^{\alpha} \circ \map{E}^{\beta} (\rho) = \map{E}^{\eta}
  (\rho)$ with $\eta = \alpha\beta$. Denote with $\hilb{H}_1$ and
  $\hilb{H}_2$ the ancillary Fock spaces of channels $\map{E}^\alpha$
  and $\map{E}^\beta$ respectively, namely $\map{E}^{\alpha}(\rho) :=
  \Tr_1[ \map{B}_{01}^{\alpha} (\rho_0 \otimes \ket{0}\bra{0}_1) ]$
  and $\map{E}^{\beta}(\rho) := \Tr_2[ \map{B}_{02}^{\beta} (\rho_0
    \otimes \ket{0}\bra{0}_2) ]$. We want to prove that for any
  $\alpha$ and $\beta$ one has the following equivalence between
  channels:
  \begin{align}\label{eq:eqchannel1}
    \Tr_{12}[(\map{B}_{02}^{\beta} \otimes \map{I}_1)
      (\map{B}_{01}^{\alpha} \otimes \map{I}_2) (\rho_0 \otimes
      \ket{00}\bra{00}_{12})] \\ = \Tr_1[(\map{B}_{01}^{\eta} (\rho_0
      \otimes \ket{0}\bra{0}_1)],\nonumber
  \end{align}
  namely
  \begin{align}\label{eq:eqchannel2}
    \begin{aligned}
      \Qcircuit @C=\qcwidth @R=\qcheight { \prepareC{\ket{0}} &
        \ustick{\hilb{H}_1} \qw & \multigate{1}{\map{B}_{01}^{\alpha}}
        & \qw & \measureD{I} \\ & \ustick{\hilb{H}_0} \qw &
        \ghost{\map{B}_{01}^{\alpha}} &
        \multigate{1}{\map{B}_{02}^{\beta}} & \qw
        \\ \prepareC{\ket{0}} & \ustick{\hilb{H}_2} \qw & \qw &
        \ghost{\map{B}_{02}^{\beta}} & \measureD{I} }
    \end{aligned}
    =
    \begin{aligned}
      \Qcircuit @C=\qcwidth @R=\qcheight { & \ustick{\hilb{H}_0} \qw &
        \multigate{1}{\map{B}_{01}^{\eta}} & \qw \\ \prepareC{\ket{0}}
        & \ustick{\hilb{H}_1} \qw & \ghost{\map{B}_{01}^{\eta}} &
        \measureD{I}}
    \end{aligned}.
  \end{align}

  By direct computation one obtains~\cite{note:decomposition}
  \begin{align}\label{eq:replacement1}
    (\map{B}_{02}^{\beta} \otimes \map{I}_1) (\map{B}_{01}^{\alpha}
    \otimes \map{I}_2) = (\map{B}_{12}^{\gamma\dagger} \otimes
    \map{I}_0) (\map{B}_{01}^{\eta} \otimes \map{I}_2)
    (\map{B}_{12}^{\delta} \otimes \map{I}_0),
  \end{align}
  where $\gamma = (1 - \alpha)/(1 - \alpha\beta)$ and $\delta =
  \gamma\beta$, namely
  \begin{align}\label{eq:replacement2}
    \begin{aligned}
      \Qcircuit @C=\qcwidth @R=\qcheight { & \ustick{\hilb{H}_1} \qw &
        \multigate{1}{\map{B}_{01}^{\alpha}} & \qw & \qw \\ &
        \ustick{\hilb{H}_0} \qw & \ghost{\map{B}_{01}^{\alpha}} &
        \multigate{1}{\map{B}_{02}^{\beta}} & \qw \\ &
        \ustick{\hilb{H}_2} \qw & \qw & \ghost{\map{B}_{02}^{\beta}} &
        \qw }
    \end{aligned}
    =
    \begin{aligned}
      \Qcircuit @C=\qcwidth @R=\qcheight { & \ustick{\hilb{H}_0} \qw &
        \qw & \multigate{1}{\map{B}_{01}^{\eta}} & \qw & \qw \\ &
        \ustick{\hilb{H}_1} \qw &
        \multigate{1}{\map{B}_{12}^{\gamma\dagger}} &
        \ghost{\map{B}_{01}^{\eta}} &
        \multigate{1}{\map{B}_{12}^{\delta}} & \qw \\ &
        \ustick{\hilb{H}_2} \qw & \ghost{\map{B}_{12}^{\gamma\dagger}}
        & \qw & \ghost{\map{B}_{12}^{\delta}} & \qw }
    \end{aligned}
  \end{align}
  Notice that for clarity the ordering of Fock spaces $\hilb{H}_0$ and
  $\hilb{H}_1$ is exchanged in the two sides of
  Eq.~\eqref{eq:replacement2}.

  Upon replacing Eq.~\eqref{eq:replacement1} into
  Eq.~\eqref{eq:eqchannel1} [or equivalently
    Eq.~\eqref{eq:replacement2} into Eq.~\eqref{eq:eqchannel2}] and
  noticing that $\map{B}_{12}^{\gamma\dagger} \ket{00}_{12} =
  \ket{00}_{12}$, one gets the statement by unitarily invariance of
  trace.
\end{proof}

Lemma~\ref{thm:composition} allows to absorb channel
$\map{E}_{\bar\alpha}$ in Eq.~\eqref{eq:setup2} into the definition of
channels $\map{E}_{\bar\beta}$ and $\map{E}_{\bar\gamma}$, so
Eq.~\eqref{eq:setup2} becomes
\begin{align}\label{eq:setup3}
  \begin{aligned}
    \Qcircuit @C=\qcwidth @R=\qcheight { \multiprepareC{1}{\rho} &
      \ustick{\hilb{H}} \qw & \gate{\map{E}^{\bar\beta'_i}} &
      \gate{\map{U}_i} & \multigate{1}{\map{E}^{\bar\gamma'}} &
      \multimeasureD{1}{\Pi}\\ \pureghost{\rho} & \ustick{\hilb{K}}
      \qw & \qw & \qw & \ghost{\map{E}^{\bar\gamma'}} & \ghost{\Pi}
      \gategroup{1}{3}{1}{4}{2mm}{--} \gategroup{1}{5}{2}{6}{2mm}{--}}
  \end{aligned}.
\end{align}

To further simplify the problem we need to be more specific about the
way information is encoded. Information can be stored in (i) quantum
efficiencies $\bar\beta'_i$'s, and (ii) in unitary channels
$\map{U}_i$'s. In the former case (i) logical $d$-it $i$ is encoded
into $\map{U} \circ \map{E}^{\bar\beta'_i}$, while in the latter case
(ii) logical $d$-it $i$ is encoded into $\map{U}_i \circ
\map{E}^{\bar\beta'}$. Clearly, also a combination of the two
encodings is possible. An example of the former case (i) is {\em ASK
  encoding}, where logical $d$-it $i$ is encoded into a single-mode
lossy optical channel $\map{E}^{\eta_i}$. An example of the latter
case (ii) is {\em PSK encoding}, where logical $d$-it $i$ is encoded
into a lossy $\phi_i$-phase shifter $\map{P}^{\phi_i} \circ
\map{E}^{\eta}$. In this work we address the problem of quantum
reading in the presence of loss with encoding of type (ii), and we
have the following Lemma.

\begin{lmm}[Commutation]
  \label{thm:commutation}
  Given an Hilbert space $\hilb{H}$, for any lossy channel
  $\map{E}^{\bar\eta} : \lin{\hilb{H}} \to \lin{\hilb{H}}$ such that
  $\bar\eta = (\eta, \dots \eta)$ is a constant vector and for any
  unitary linear and passive channel $\map{U} : \lin{\hilb{H}} \to
  \lin{\hilb{H}}$, one has that $\map{E}^{\bar\eta}$ commutes with
  $\map{U}$, namely
  \begin{align}\label{eq:eqchannel3}
   \map{U} \circ \map{E}^{\bar\eta} (\rho) = \map{E}^{\bar\eta} \circ
   \map{U} (\rho), \qquad \forall \rho \in \lin{\hilb{H}}.
  \end{align}
\end{lmm}

\begin{proof}
  Denote with $\hilb{K} = \bigotimes_i \hilb{K}_i$ the ancillary
  Hilbert space in the definition of lossy channel
  $\map{E}^{\bar\eta}$, namely $\map{E}^{\bar\eta}(\rho) =
  \Tr_B[\bigotimes_i \map{B}_i^{\eta} (\rho \otimes \sigma_\hilb{K})]$
  where $\sigma_{\hilb{K}} = \bigotimes_i \ket{0}\bra{0}_i \in
  \lin{\hilb{K}}$ and beamsplitter $\map{B}_i^{\eta}$ acts on Fock
  spaces $\hilb{H}_i$ and $\hilb{K}_i$. Then the statement can be
  reformulated as
  \begin{align}\label{eq:eqchannel4}
    \begin{aligned}
      \Qcircuit @C=\qcwidth @R=\qcheight { & \ustick{\hilb{H}_0} \qw &
        \multigate{1}{\map{B}_0^{\eta}} & \qw & \multigate{3}{\map{U}}
        & \qw \\ \prepareC{\ket{0}} & \ustick{\hilb{K}_0} \qw &
        \ghost{\map{B}_0^{\eta}} & \measureD{I} &
        \pureghost{\map{U}}\\ & & & \lstick{\dots} & \ghost{\map{U}} &
        \qw \\ & \ustick{\hilb{H}_m} \qw &
        \multigate{1}{\map{B}_m^{\eta}} & \qw & \ghost{\map{U}} & \qw
        \\ \prepareC{\ket{0}} & \ustick{\hilb{K}_m} \qw &
        \ghost{\map{B}_m^{\eta}} & \measureD{I} & \pureghost{\map{U}}}
    \end{aligned}
    =
    \begin{aligned}
      \Qcircuit @C=\qcwidth @R=\qcheight { & \ustick{\hilb{H}_0} \qw &
        \multigate{3}{\map{U}} & \qw & \qw &
        \multigate{1}{\map{B}_0^{\eta}} & \qw \\ & &
        \pureghost{\map{U}} & \prepareC{\ket{0}} & \ustick{\hilb{K}_0}
        \qw & \ghost{\map{B}_0^{\eta}} & \measureD{I} \\ & \qw &
        \ghost{\map{U}} & \rstick{\dots} \qw \\ & \ustick{\hilb{H}_m}
        \qw & \ghost{\map{U}} & \qw & \qw &
        \multigate{1}{\map{B}_m^{\eta}} & \qw \\ & &
        \pureghost{\map{U}} & \prepareC{\ket{0}} & \ustick{\hilb{K}_m}
        \qw & \ghost{\map{B}_m^{\eta}} & \measureD{I} }
    \end{aligned}
  \end{align}

  It follows by direct computation that
  \begin{align}\label{eq:replacement3}
    \bigotimes_i \map{B}_i^{\eta} (\map{U} \otimes \map{U}_{\hilb{K}}) =
    (\map{U} \otimes \map{U}_{\hilb{K}}) \bigotimes_i
    \map{B}_i^{\eta},
  \end{align}
  where $\map{U}_{\hilb{K}}$ denotes channel $\map{U}$ acting on
  Hilbert space $\hilb{K}$, or equivalently
  \begin{align}\label{eq:replacement4}
    \begin{aligned}
      \Qcircuit @C=\qcwidth @R=\qcheight { & \ustick{\hilb{H}_m} \qw &
        \multigate{2}{\map{U}} & \qw & \qw &
        \multigate{5}{\map{B}_m^{\eta}} & \qw \\ & \lstick{\dots} &
        \ghost{\map{U}} & \rstick{\dots} \qw \\ & \ustick{\hilb{H}_0}
        \qw & \ghost{\map{U}} & \multigate{1}{\map{B}_0^{\eta}} & \qw
        \\ & \ustick{\hilb{K}_0} \qw &
        \multigate{2}{\map{U}_{\hilb{K}}} & \ghost{\map{B}_0^{\eta}} &
        \qw \\ & \lstick{\dots} & \ghost{\map{U}_{\hilb{K}}} &
        \rstick{\dots} \qw \\ & \ustick{\hilb{K}_m} \qw &
        \ghost{\map{U}_{\hilb{K}}} & \qw & \qw &
        \ghost{\map{B}_m^{\eta}} & \qw }
    \end{aligned}
    =
    \begin{aligned}
      \Qcircuit @C=\qcwidth @R=\qcheight { & \ustick{\hilb{H}_m} \qw &
        \multigate{5}{\map{B}_m^{\eta}} & \qw & \qw & \qw &
        \multigate{2}{\map{U}} & \qw \\ & &
        \pureghost{\map{B}_m^{\eta}} & & & \lstick{\dots} &
        \ghost{\map{U}} & \qw \\ & & \pureghost{\map{B}_m^{\eta}} & &
        \ustick{\hilb{H}_0} \qw & \multigate{1}{\map{B}_0^{\eta}} &
        \ghost{\map{U}} & \qw \\ & & \pureghost{\map{B}_m^{\eta}} & &
        \ustick{\hilb{K}_0} \qw & \ghost{\map{B}_0^{\eta}} &
        \multigate{2}{\map{U}_{\hilb{K}}} & \qw \\ & &
        \pureghost{\map{B}_m^{\eta}} & & & \lstick{\dots} &
        \ghost{\map{U}_{\hilb{K}}} & \qw \\ & \ustick{\hilb{K}_m} \qw
        & \ghost{\map{B}_m^{\eta}} & \qw & \qw & \qw &
        \ghost{\map{U}_{\hilb{K}}} & \qw \\ }
    \end{aligned}
  \end{align}

  Upon replacing Eq.~\eqref{eq:replacement3} into
  Eq.~\eqref{eq:eqchannel3} [or equivalently
    Eq.~\eqref{eq:replacement4} into Eq.~\eqref{eq:eqchannel4}] the
  statement follows, namely
  \begin{align*}
    \map{U} \circ \map{E}^{\bar\eta} (\rho) = & \Tr_{\hilb{K}} [
      (\map{U} \otimes I_{\hilb{K}}) \bigotimes_i \map{B}_i^{\eta}
      (\rho \otimes \sigma_{\hilb{K}}) ] = \\ & \Tr_{\hilb{K}} [
      (\map{U} \otimes \map{U}_{\hilb{K}}) \bigotimes_i
      \map{B}_i^{\eta} (\rho \otimes \sigma_{\hilb{K}}) ] = \\ &
    \Tr_{\hilb{K}} [ \bigotimes_i \map{B}_i^{\eta} (\map{U} \otimes
      \map{U}_{\hilb{K}}) (\rho \otimes \sigma_{\hilb{K}})] = \\ &
    \Tr_{\hilb{K}} [ \bigotimes_i \map{B}_i^{\eta} (\map{U} \otimes
      \map{I}_{\hilb{K}}) (\rho \otimes \sigma_{\hilb{K}}) ] =
    \map{E}^{\bar\eta} \circ \map{U} (\rho),
  \end{align*}
  where second equality holds due to unitarily invariance of trace and
  second-to-last holds since $\sigma_{\hilb{K}} = \bigotimes_i
  \ket{0}\bra{0}_i$.
\end{proof}

Whenever the hypothesis of Lemma~\ref{thm:commutation} is satisfied,
namely losses affecting each mode of unitary channel $\map{U}$ are
equal - this is true for example when modes are implemented by
optical-fiber paths of the same length - applying
Lemma~\ref{thm:commutation} and Lemma~\ref{thm:composition} again the
quantum reading strategy in Eq.~\eqref{eq:setup3} can be rewritten as
\begin{align}\label{eq:setup4}
  \begin{aligned}
    \Qcircuit @C=\qcwidth @R=\qcheight { \multiprepareC{1}{\rho} &
      \ustick{\hilb{H}} \qw & \gate{\map{U}_i} &
      \multigate{1}{\map{E}_{\bar\eta}} &
      \multimeasureD{1}{\Pi}\\ \pureghost{\rho} & \ustick{\hilb{K}}
      \qw & \qw & \ghost{\map{E}_{\bar\eta}} & \ghost{\Pi}
      \gategroup{1}{4}{2}{5}{2mm}{--}}
  \end{aligned}.
\end{align}
We proved that when the information is not encoded in the efficiency
of each memory cell, quantum reading of lossy optical devices can be
recasted to the discrimination~\cite{BDD11,DBDMJD12,Aci01,DLP02,DFY07}
of lossless optical devices with a lossy detector.

\section{Gaussian quantum reading}
\label{sect:gaussread}

In this Section we derive the Standard Quantum Limit - namely, the
optimal strategy using coherent states and homodyne measurements - and
an optimal Gaussian strategy - using squeezed coherent states and
homodyne measurements - for quantum reading with PSK encoding for any
value of the phase difference between signals encoding logical $0$ and
$1$ and in the presence of loss. Let us first fix the notation
\cite{SZ97,GK04,VW06,BR02}.

A {\em coherent state} $\ket{\alpha}$ is obtained by applying the
displacement operator $D(\alpha) := e^{\alpha a^\dagger - \alpha^* a}$
to the vacuum state $\ket{0}$, namely $\ket{\alpha} :=
D(\alpha)\ket{0}$. The energy of a coherent state is given by
$E(\alpha) = |\alpha|^2$. A {\em squeezed coherent state}~\cite{Yue76}
$\ket{\alpha,\xi}$ is obtained by subsequently applying the squeezing
operator $S(\xi) := e^{\frac12(\xi^*a^2 - \xi a^{\dagger2})}$ and the
displacement operator $D({\alpha})$ to the vacuum state $\ket{0}$,
namely $\ket{\alpha,\xi} := D(\alpha)S(\xi)\ket{0}$. The energy of a
squeezed coherent state is given by $E(\alpha,\phi) = |\alpha|^2 +
\sinh^2|\xi|$.

For any state $\rho$ it is useful to introduce its Wigner function
$W_{\rho}(x,p)$ defined as
\begin{align*}
  W_{\rho}(x,p) := \frac{1}{\pi} \int_{-\infty}^{\infty}
  \bra{x-y,\psi} \rho \ket{x+y,\psi} e^{-i2py} dy,
\end{align*}
so the Wigner function of squeezed coherent state
$\ket{e^{i\phi}a,e^{i\theta}r}$ is given by
\begin{align*}
  W_{\ket{\alpha,\xi}}(x,p) = \frac{1}{\pi} e^{-(x'^2+p'^2)}.
\end{align*}
where
\begin{align*}
  \left( \begin{array}{c} x' \\ p' \end{array} \right) =
  \left( \begin{array}{cc} e^{-r} \cos\frac\theta2 & \sin\frac\theta2
    \\ -\sin\frac\theta2 & e^{r} \cos\frac\theta2 \end{array} \right)
  \left( \begin{array}{c}x\\p\end{array} \right) -
    \left( \begin{array}{c} a \cos\phi \\ a \sin\phi \end{array}
    \right),
\end{align*}
and that of coherent state $\ket{e^{i\phi}a}$ can be obtained as a
particular case by setting $r = 0$.

The POVM $\Pi^{\psi}$ describing an {\em homodyne
  measurement}~\cite{BW97,TS04} along quadrature $X(\psi) :=
\frac{1}{\sqrt{2}}(e^{i\psi}a + e^{-i\psi}a^\dagger)$ associates to
state $\rho$ the probability distribution $p(x|\rho,\psi)$ given by
\begin{align*}
  p(x|\rho,\psi) = \int_{-\infty}^{\infty} dp W_{\rho}(x',p'),
\end{align*}
where
\begin{align*}
  \left( \begin{array}{c} x' \\ p' \end{array} \right) =
  \left( \begin{array}{cc} \cos\psi & -\sin\psi \\ \sin\psi &
    \cos\psi \end{array} \right)
  \left( \begin{array}{c}x\\p\end{array} \right).
\end{align*}

In the following we denote with $p(x|\alpha,\xi,\psi)$ the conditional
probability distribution of outcome $x$ of homodyne measurement
$\Pi^{\psi}$ with efficiency $\eta$ given input squeezed coherent
state $\ket{\alpha,\xi}$ with $\alpha = e^{i\phi}a$ and $\xi =
e^{i\theta}r$, given by the Gaussian
\begin{align}\label{eq:pout}
  p(x|\alpha,\xi,\psi) := {\sqrt\frac{1}{\pi\sigma^2(\xi,\psi)}}
  e^{-\frac{(x-x_0(\alpha,\psi))^2}{\sigma^2(\xi,\psi)}},
\end{align}
where
\begin{align*}
  & x_0(\alpha,\psi) := a \cos(\psi-\phi),\\ & \sigma(\xi,\psi) :=
  \sqrt{ e^{-2r}\cos^2(\psi-\frac\theta2) +
    e^{2r}\sin^2(\psi-\frac\theta2) + \frac{1-\eta}{4\eta}},
\end{align*}
while the analogous distribution for coherent state $\ket{e^{i\phi}a}$
can be obtained as a particular case by setting $r = 0$. The term
$\frac{1-\eta}{4\eta}$ in the definition of $\sigma(\xi,\psi)$ is due
to the fact that the conditional probability distribution in the
presence of loss is the convolution of the ideal conditional
probability distribution with a Gaussian with variance
$\frac{1-\eta}{4\eta}$ (see Refs.~\cite{DPS00,PDS01,DMP96,PR04}).

We introduce now the problem of quantum reading with PSK encoding with
a lossy source of squeezed coherent states
$\map{E}^{\alpha}(\ket{\alpha,\xi}\bra{\alpha,\xi})$ and lossy
homodyne measurement $\map{E}^{\gamma\vee}(\Pi^{\psi})$. With the same
notation as in Definition~\ref{def:qread}, we assume that the (binary)
information is encoded into lossy phase shifters $\map{P}^{\delta_i}
\circ \map{E}^{\beta}$ with $i=0,1$, and that no prior information is
provided, namely $p_0=p_1=1/2$. Since homodyne measurement has
infinitely many outcomes, we introduce a classical postprocessing
$\map{J}$ (classical wires being denoted with $\Qcircuit @C=0.7em
@R=1em {& \cw}$) outputting binary outcome $j$, so that $j$ is our
guess for the value $i$ of the bit encoded into the unknown phase
shifter. Then the strategy in Eq.~\eqref{eq:setup2} becomes
\begin{align}\label{eq:setup5}
  \begin{aligned}
    \Qcircuit @C=\qcwidth @R=\qcheight { \prepareC{\ket{\alpha,\xi}} &
      \gate{\map{E}^{\alpha}} & \gate{\map{E}^{\beta}} &
      \gate{\map{P}^{\delta_i}} & \gate{\map{E}^{\gamma}} &
      \measureD{\Pi^\psi} & \cgate{\map{J}} & \cw
      \gategroup{1}{1}{1}{2}{2mm}{--} \gategroup{1}{3}{1}{4}{2mm}{--}
      \gategroup{1}{5}{1}{7}{2mm}{--}}
  \end{aligned}.
\end{align}
Clearly, for fixed probe Gaussian quantum reading reduces to Gaussian
state discrimination
\cite{Dol73,Hel76,TFFIFTS11,OP03,OBH96,MUWTMAL12}.

Due to Lemmas~\ref{thm:composition} and~\ref{thm:commutation} lossy
channels can be absorbed in the definition of lossy POVM, namely the
strategy in Eq.~\eqref{eq:setup5} can be rewritten as
\begin{align}\label{eq:setup6}
  \begin{aligned}
    \Qcircuit @C=\qcwidth @R=\qcheight { \prepareC{\ket{\alpha,\xi}} &
      \gate{\map{P}^{\delta_i}} & \gate{\map{E}^{\eta}} &
      \measureD{\Pi^\psi} & \cgate{\map{J}} & \cw
      \gategroup{1}{3}{1}{5}{2mm}{--}}
  \end{aligned},
\end{align}
with $\eta=\alpha\beta\gamma$. Notice that it is not restrictive to
assume $\delta_0 = 0$, since a phase shifter $\map{P}^{-\delta_0}$ can
be reabsorbed in the definition of input state or homodyne
measurement, finally leading to the discrimination of phase shifters
$\map{P}^{\delta_0} = \map{I}$ and $\map{P}^{\delta_1} =
\map{P}^{\delta}$, where $\delta = -\delta_0\delta_1$, with $\delta
\in [0,\pi]$. Notice that the application of phase shifter
$\map{P}^\delta$ to a squeezed coherent state $\ket{\alpha,\xi}$
gives~\cite{note:psscs} the squeezed coherent state $\ket{
  e^{i\delta}\alpha, e^{i2\delta}\xi}$. Any classical postprocessing
of outcome $x$ can be described by a function $q(j|x)$ that evaluates
to $1$ if one guesses $j$ from outcome $x$ and to $0$ otherwise, so
the probability of error is given by
\begin{align}\label{eq:perr}
  P_e = \frac12 \int dx \; q(1|x) p(x|\alpha,\xi\,\psi) + q(0|x)
  p(x|e^{i\delta}\alpha,e^{i2\delta}\xi,\psi).
\end{align}

In the following, we denote with $\erf(x) := 2/\sqrt{\pi} \int_{0}^x
dt \; e^{-t^2}$ the error function and with $\Omega(x)$ the unit step
function that evaluates to $1$ if $x \ge 0$ and to $0$ otherwise.

The Standard Quantum Limit for quantum reading with PSK encoding in
the ideal case was derived in Ref.~\cite{BDD11}. Here we generalize
that result to the lossy case.

\begin{thm}[Standard Quantum Limit]
  \label{thm:sql}
  For any energy threshold $E$, any efficiency $\eta$ and any phase
  $\delta$, the optimal coherent state $\ket{\alpha^*}$ with $\alpha^*
  = e^{i\phi^*} a^*$ and the optimal homodyne measurement
  $\Pi^{\psi^*}$ with efficiency $\eta$ for Gaussian quantum reading
  of phase shifters $\{\map{I}, \map{P}^\delta\}$ with energy
  $E(\alpha) \le E$ are given by $\phi^* = -\delta/2$, $a^* =
  \sqrt{E}$, and $\psi^* = \pi/2$. The optimal tradeoff is given by
  $P_e = \frac12\left[1+\erf\left( \frac{x_0(\alpha^*,\psi^*)}
    {\sigma(0,\psi^*)}\right)\right]$.
\end{thm}

\begin{proof}
  The minimum of $P_e$ in Eq.~\eqref{eq:perr} is attained when
  \begin{align*}
    q(y|x) = \Omega[(-)^y (p(x|\alpha,0,\psi) -
      p(x|e^{i\delta}\alpha,0,\psi))]
  \end{align*}
  and thus $P_e$ is given by the overlap between the Gaussians
  $p(x|\alpha,0,\psi)$ and $p(x|e^{i\delta}\alpha,0,\psi)$, namely
  \begin{align*}
    P_e = \frac12 \left[1 + \int_{\set{A}} dx \left(
      p(x|e^{i\delta}\alpha,0,\psi) - p(x|\alpha,0,\psi) \right)
      \right],
  \end{align*}
  where $\set{A} := \{x | p(x|\alpha,0,\psi) \ge
  p(x|e^{i\delta}\alpha,0,\psi)\}$.

  From Eq.~\eqref{eq:pout} it follows that $P_e$ depends on the phases
  $\phi$ and $\psi$ only through the sum $\psi-\phi$, so without loss
  of generality we fix $\psi = \psi^*$, namely homodyne measurement is
  performed along quadrature $P$.

  Notice that for any coherent state $\ket{e^{i\phi}a}$ such that
  $E(e^{i\phi}a) \le E$, one has that the state $\ket{e^{i\phi}a^*}$
  is such that $E(e^{i\phi}a^*) = E$ and $P_e(e^{i\phi}a^*) \le
  P_e(e^{i\phi}a)$. Indeed, while $\sigma(0,\psi^*)$ does not depend
  on $a$, one has that $|x_0(e^{i(\delta+\phi)}a^*,\psi^*) -
  x_0(e^{i\phi}a^*,\psi^*)| \ge |x_0(e^{i(\delta+\phi)}a,\psi^*) -
  x_0(e^{i\phi}a,\psi^*)|$. Then the optimal value for parameter $a$
  is $a^*$.
  
  Notice that for any coherent state $\ket{e^{i\phi}a^*}$ such that
  $E(e^{i\phi}a^*) = E$, one has that the state $\ket{e^{i\phi^*}a^*}$
  is such that $E(e^{i\psi^*}a^*) = E$ and $P_e(e^{i\phi^*}a^*) \le
  P(e^{i\phi}a)$. Indeed, while $\sigma(0,\psi^*)$ does not depend on
  $a$, one has that $|x_0(e^{i(\delta+\phi^*)}a^*,\psi^*) -
  x_0(e^{i\phi^*}a^*,\psi^*)| \ge |x_0(e^{i(\delta+\phi)}a^*,\psi^*) -
  x_0(e^{i\phi}a^*,\psi^*)|$. Then the optimal value for phase $\phi$
  is $\phi^*$.
\end{proof}

We can now introduce the main result of this work, namely an optimal
strategy for quantum reading with PSK encoding with squeezed coherent
states and homodyne measurements.

\begin{thm}[Optimal Gaussian quantum reading]\label{thm:gaussopt}
  For any energy threshold $E$, any efficiency $\eta$ and any phase
  $\delta$, the optimal squeezed coherent state $\ket{\alpha^*,\xi^*}$
  with $\alpha^* = e^{i\phi^*} a^*$ and $\xi^* = e^{i\theta^*} r^*$
  and the optimal homodyne measurement $\Pi^{\psi^*}$ with efficiency
  $\eta$ for Gaussian quantum reading of phase shifters $\{I,
  P_\delta\}$ with energy $E(\alpha) \le E$ are given by $\phi^* =
  -\delta/2$, $a^* = \sqrt{E-\sinh^2r^*}$, $\psi^* = \pi/2$. Whenever
  $\theta^* = -\delta -\pi\Omega(\frac{\pi}{2}-\delta)$ one has
  \begin{align*}
    r^* = \frac12 \ln \frac{
      \sqrt{\left(2E+1+\frac{1-\eta}{4\eta}\right)^2-4E(E+1)\cos^2\theta^*}
      + \cos\theta^*} {(\cos\theta^*+1)(2E+1)+\frac{1-\eta}{4\eta}},
  \end{align*}
  and the optimal tradeoff is given by $P_e =
  \frac12\left[1+\erf\left(\frac{x_0(\alpha^*,\psi^*)}
    {\sigma(\xi^*,\psi^*)}\right)\right]$.
\end{thm}

\begin{proof}
  The minimum of $P_e$ in Eq.~\eqref{eq:perr} is attained when
  \begin{align*}
    q(y|x) = \Omega[(-)^y (p(x|\alpha,\xi,\psi) -
      p(x|e^{i\delta}\alpha,e^{i2\delta}\xi,\psi))]
  \end{align*}
  and thus $P_e$ is given by the overlap between the Gaussians
  $p(x|\alpha,\xi,\psi)$ and
  $p(x|e^{i\delta}\alpha,e^{i2\delta}\xi,\psi)$, namely
  \begin{align}\label{eq:gaussperr}
    P_e = \frac12 \left[1 + \int_{\set{A}} dx \left(
      p(x|e^{i\delta}\alpha,e^{i2\delta}\xi,\psi) -
      p(x|\alpha,\xi,\psi) \right) \right],
  \end{align}
  where $\set{A} := \{x | p(x|\alpha,\xi,\psi) \ge
  p(x|e^{i\delta}\alpha,e^{i2\delta}\xi,\psi)\}$.

  From Eq.~\eqref{eq:pout} it follows that $P_e$ depends on the phases
  $\phi$, $\theta$ and $\psi$ only through the sums $\psi-\phi$ and
  $\psi-\frac\theta2$, so without loss of generality we fix $\psi =
  \psi^*$, namely homodyne measurement is performed along quadrature
  $P$.

  Notice that for any squeezed coherent state $\ket{e^{i\phi}a,\xi}$
  such that $E(e^{i\phi}a,\xi) \le E$, one has that the state
  $\ket{e^{i\phi}a(r),\xi}$ with $a(r) := \sqrt{E-\sinh^2r}$ is such
  that $E(e^{i\phi}a(r),\xi) = E$ and $P_e(e^{i\phi}a(r),\xi) \le
  P_e(e^{i\phi}a,\xi)$. Indeed, while $\sigma(\xi,\psi^*)$ does not
  depend on $a$, one has that $|x_0(e^{i(\delta+\phi)}a(r),\psi^*) -
  x_0(e^{i\phi}a(r),\psi^*)| \ge |x_0(e^{i(\delta+\phi)}a,\psi^*) -
  x_0(e^{i\phi}a,\psi^*)|$. Then constraint $E(\alpha,\xi) \le E$ can
  be recasted without loss of generality into $E(\alpha,\xi) = E$,
  allowing to eliminate parameter $a$ by writing $a = a(r)$.
  
  Notice that for any squeezed coherent state
  $\ket{e^{i\phi}a(r),\xi}$ such that $E(e^{i\phi}a(r),\xi) = E$, one
  has that the state $\ket{e^{i\phi^*}a(r),\xi}$ is such that
  $E(e^{i\phi^*}a(r),\xi) = E$ and $P_e(e^{i\phi^*}a(r),\xi) \le
  P(e^{i\phi}a(r),\xi)$. Indeed, while $\sigma(e^{i\theta}r,\psi^*)$
  does not depend on $\phi$, one has that
  $|x_0(e^{i(\delta+\phi^*)}a(r),\psi^*) -
  x_0(e^{i\phi^*}a(r),\psi^*)| \ge |x_0(e^{i(\delta+\phi)}a(r),\psi^*)
  - x_0(e^{i\phi}a,\psi^*)|$. Then the optimal value for phase $\phi$
  is $\phi^*$.

  When $\theta = \theta^*$, an explicit evaluation of
  Eq. \eqref{eq:gaussperr} leads to $P_e = \frac12\left[1 +
    \erf(\frac{ x_0(e^{i\phi^*}a(r),\psi^*) } {
      \sigma(e^{i\theta^*}r,\psi^*) })\right]$, and the optimal value
  for parameter $r$ can be obtained minimizing $P_e$. Since $\erf(x)$
  is a monotone increasing function in $x$, minimizing $P_e$ is
  equivalent to minimizing $\frac{ x_0(e^{i\phi^*}a(r),\psi^*) } {
    \sigma(e^{i\theta^*}r,\psi^*) }$. It is lengthy but not difficult
  to verify that the equation $\frac{\partial}{\partial r} \frac{
    x_0(e^{i\phi^*}a(r),\psi^*) } { \sigma(e^{i\theta^*}r,\psi^*) } =
  0$ admits the only solution $r = r^*$ and that $r^*$ is a minimum,
  so the statement remains proved.
\end{proof}

Notice that for $\alpha = \alpha^*$ and $\psi = \psi^*$, for any $r$
the only choices of $\theta$ such that $\sigma(\xi,\psi^*) =
\sigma(e^{i2\delta}\xi,\psi^*)$ are $\theta = -\delta$ and $\theta =
-\delta -\pi$, and the choice $\theta = \theta^*$ given by
Prop.~\ref{thm:gaussopt} corresponds to the one minimizing
$\sigma(\xi,\psi^*)$ (see next Section and
Fig.~\ref{fig:strategy}). We obtained numerical evidence that the
choice $\theta = \theta^*$ is optimal whenever (i) $\delta \ge \pi/2$
for any energy threshold $E$ and any efficiency $\eta$, or (ii) when
$\delta \le \pi/2$ for sufficiently high $E$. However, when $\delta
\le \pi/2$ for sufficiently low $E$ the choice $\theta = \theta^*$ is
not optimal anymore, and the second statement in
Prop.~\ref{thm:gaussopt} can not be applied.

When $\delta = \pi$ one has $\theta^* = -\pi$, and thus the expression
for $r^*$ in Prop.~\ref{thm:gaussopt} is not defined for $\eta =
1$. In this case the limit $\eta \to 1$ must be considered, leading to
$r^* = \arcsinh[E/\sqrt{2E+1}]$. Notice that this particular
expression for $r^*$ when $\delta = \pi$ and $\eta = 1$ is optimal
also for hybrid quantum reading with Gaussian probe and arbitrary
measurement (see Ref.~\cite{NYGSP12}).

\section{Optimal Gaussian strategy versus Standard Quantum Limit}
\label{sect:comp}

This Section is devoted to the comparison between the Standard Quantum
Limit and the optimal Gaussian strategy for quantum reading derived in
Sect.~\ref{sect:gaussread}.

Figure~\ref{fig:strategy} provides a phase-space representation of the
Wigner function of the states attaining the Standard Quantum Limit
(light gray circles) and the optimal Gaussian strategy (bold-line
ellipses) as given by Props.~\ref{thm:sql} and \ref{thm:gaussopt},
respectively.
\begin{figure}[htb]
  \includegraphics[width=1\columnwidth]{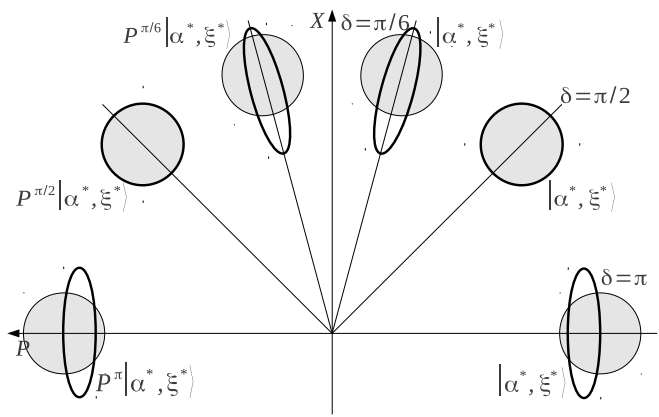}
  \caption{Phase-space representation of the Wigner function of the
    optimal Gaussian state (bold-line ellipses) $\ket{\alpha^*,\xi^*}$
    for quantum reading of a $\delta$-phase shifter for $\delta = \pi,
    \pi/2, \pi/6$, as given by Prop.~\ref{thm:gaussopt}. The optimal
    homodyne measurement is along quadrature $P$ (horizontal in
    Figure). The optimal coherent state (light gray circles) attaining
    the Standard Quantum Limit is also depicted for each value of
    $\delta$. The optimal strategy outperforms the Standard Quantum
    Limit for any value of $\delta$ except $\delta=\pi/2$, where the
    two strategies coincide. The maximal advantage over the Standard
    Quantum Limit is achieved for $\delta=\pi$ and in the regime of
    small $\delta$.}
  \label{fig:strategy}
\end{figure}
The Figure provides an intuitive understanding of the advantage given
by squeezed coherent states over coherent states. On the one hand, for
fixed energy the more two states are squeezed, the more their Wigner
functions get ``closer'' and thus hardly distinguishable. On the other
hand, when squeezing is performed approximately along the quadrature
being measured, the Wigner functions become ``thinner'' as squeezing
increases. These two phenomena are clearly contrasting, but when the
optimal tradeoff is taken, a dramatic improvement in the precision of
the discrimination is experienced, as discussed in the next
paragraphs. The only value of the phase $\delta$ for which squeezed
coherent states do not provide an advantage over coherent states is
$\delta = \pi/2$. The maximum advantage is achieved in the two regimes
$\delta \sim \pi$ and $\delta \sim 0$, since in both cases the Wigner
function of the optimal state $\rho^*$ is squeezed approximately along
the quadrature which is measured - this is rigorously true only for
$\delta = \pi$.

The advantage of using the $\delta = \pi$ encoding is obvious - it is
clearly the choice giving the lower tradeoff between energy and
probability of error (see later discussion and
Fig. \ref{fig:popt}). Different choices, and in particular the regime
$\delta \sim 0$, can be exploited in several applications for tuning
the minimum energy required by a Gaussian reader to retrieve some
information. For example, a read-only-once memory could be implemented
by triggering a device to self erase after being exposed to a
radiation of energy $E$, where $E$ is the minimum energy required to
read the memory with a given probability of error $P_e$, as given by
Prop.~\ref{thm:gaussopt}. Notice that the quantum reading problem in
the regime $\delta \sim 0$ shares analogies with the problem of
channel estimation for low-noise parameter~\cite{HKO05,HKO06}.

In Fig.~\ref{fig:ropt} the optimal squeezing parameter $r^*$ given by
Prop.~\ref{thm:gaussopt} is plotted as a function of $E$ for different
values of $\delta$ and $\eta$.
\begin{figure}[htb]
  \includegraphics[width=1\columnwidth]{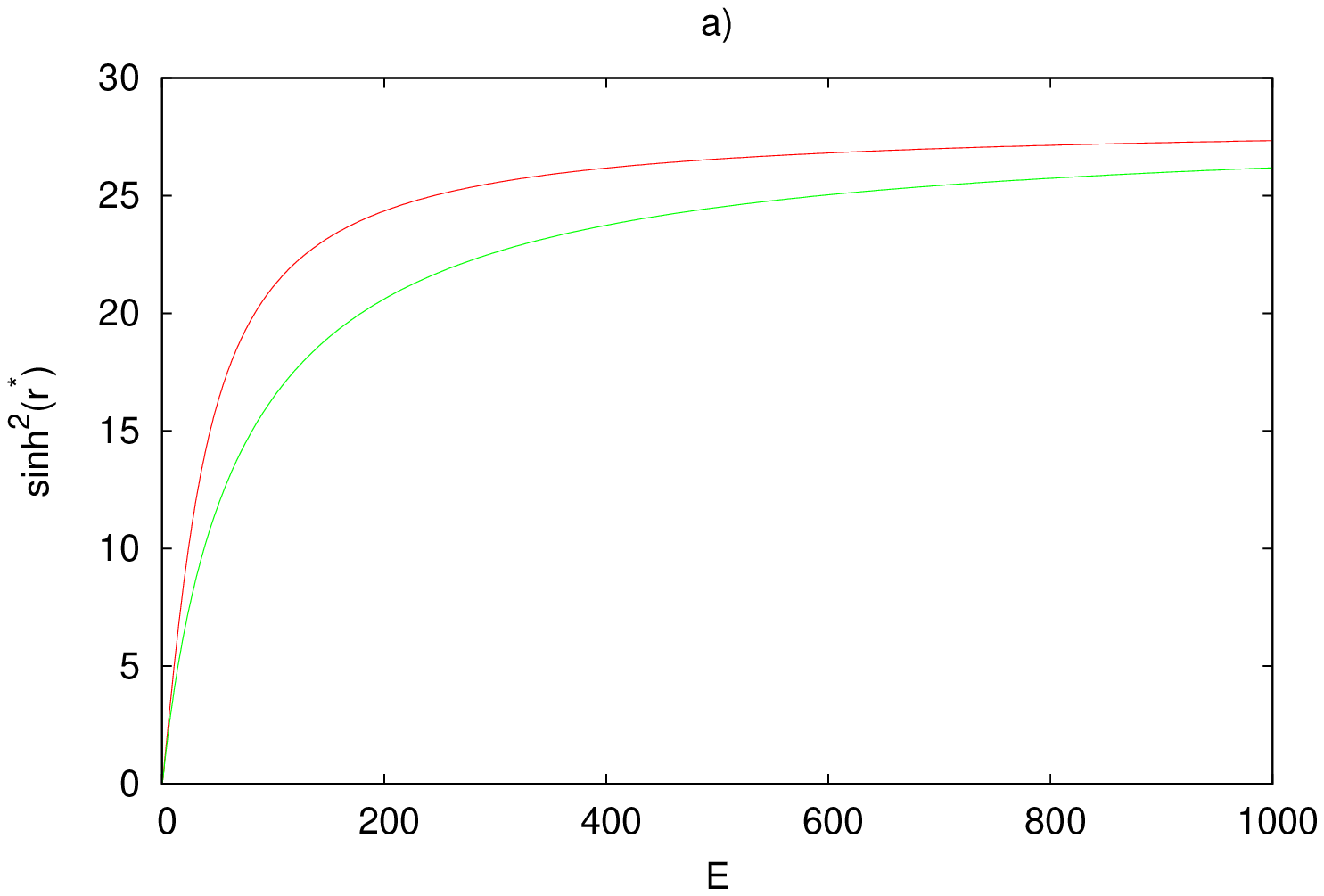}
  \includegraphics[width=1\columnwidth]{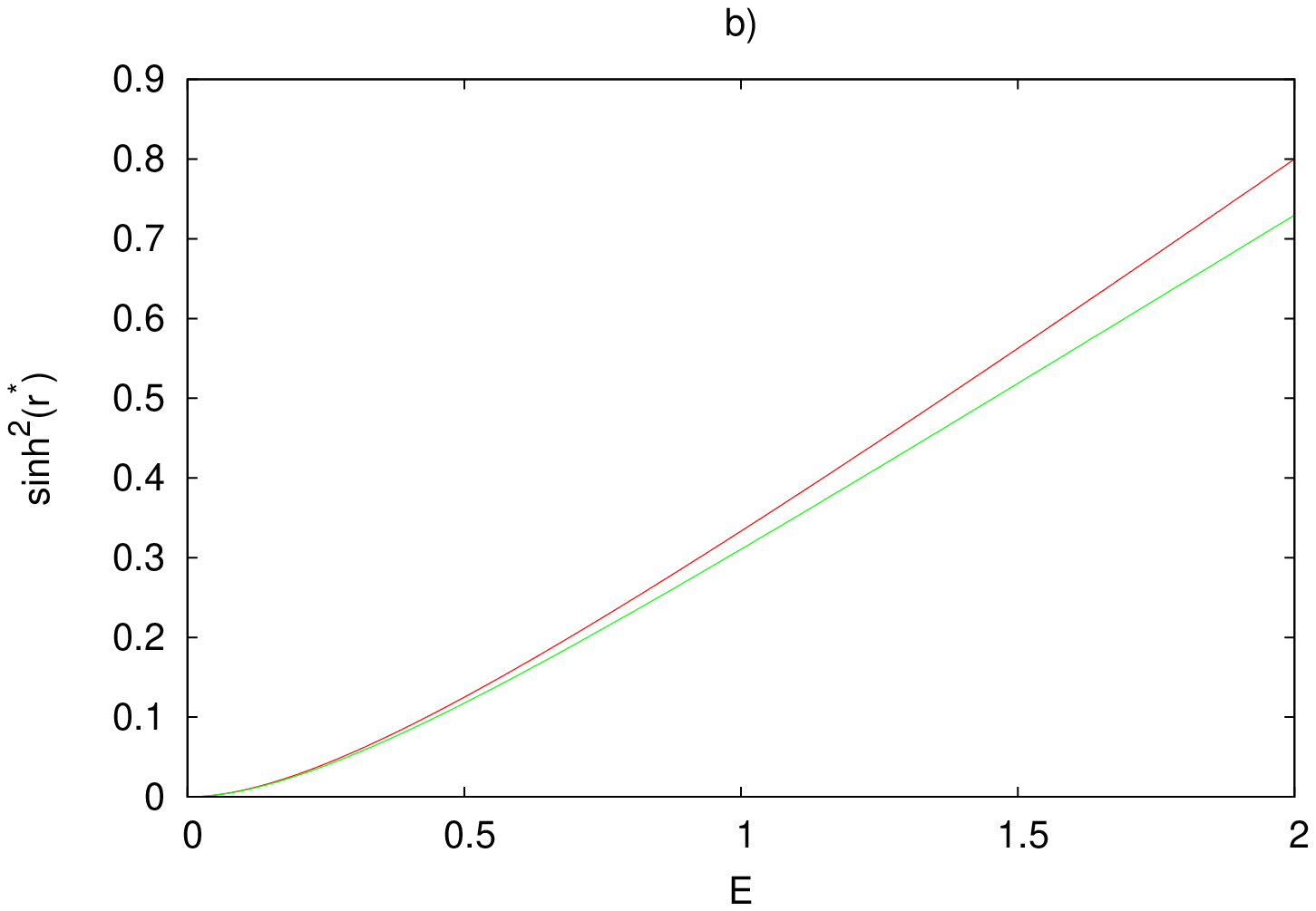}
  \caption{(Color online) Optimal value $\sinh^2(r^*)$ of squeezing
    energy for Gaussian quantum reading of phase shifters $\map{I}$
    and $\map{P}^{\delta}$ as given by Prop.~\ref{thm:gaussopt} versus
    the energy threshold $E$, for $\delta=\pi/180$ [Fig. a)] and
    $\delta=\pi$ [Fig. b)] and for detection efficiency $\eta=1$
    (upper red line in both figures) and $\eta=0.9$ (lower green line
    in both figures). Noticeably, just a tiny fraction of the total
    energy is used for squeezing in optimal Gaussian quantum reading,
    thus keeping the theoretical proposal close to the regime of
    experimental feasibility. Indeed, the generation of squeezed
    coherent states with squeezing parameter $r \sim 1.5$, namely
    squeezing energy $\sinh^2(r) \sim 4.5$, is
    reported~\cite{ESBHVMMS10,MAESVS11}.}
  \label{fig:ropt}
\end{figure}
Remarkably, even for very low values of the probability of error
$P_e$, the corresponding $r^*$ is comparable with experimentally
attainable values of the squeezing parameter, in particular in the
regime $\delta \sim \pi$. For example, when $\delta = \pi$, setting
the energy $E = 4$ and the efficiency of the homodyne measurement
$\eta = 0.9$ (an arguably fairly conservative assumption) one obtains
a probability of error as low as $P_e \sim 6.5 \cdot 10^{-9}$ with
optimal squeezing parameter given by $r^* \sim 1.0$. In
Refs.~\cite{ESBHVMMS10,MAESVS11}, the generation of squeezed coherent
states with squeezing parameter up to $r_{\textrm{dB}} \sim 12.7 dB$
is reported - namely such that $r = r_{\textrm{dB}} / (20 \log_{10}e)
\sim 1.5$.

The optimal tradeoff between energy and probability of error in
Gaussian quantum reading with PSK encoding as given by
Prop.~\ref{thm:gaussopt} is plotted in Fig.~\ref{fig:popt} for
different values of $\delta$ and $\eta$.
\begin{figure}[htb]
  \includegraphics[width=1\columnwidth]{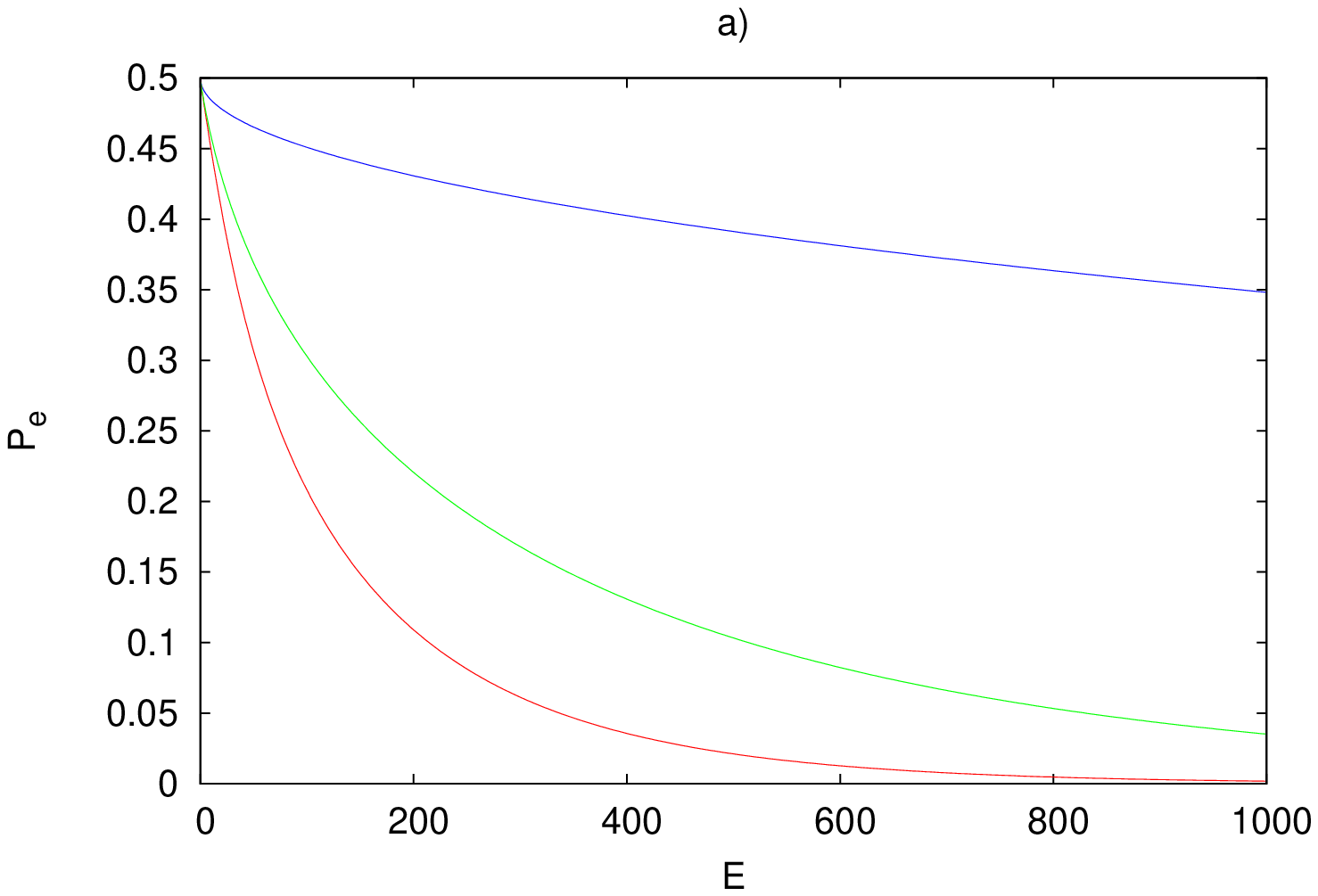}
  \includegraphics[width=1\columnwidth]{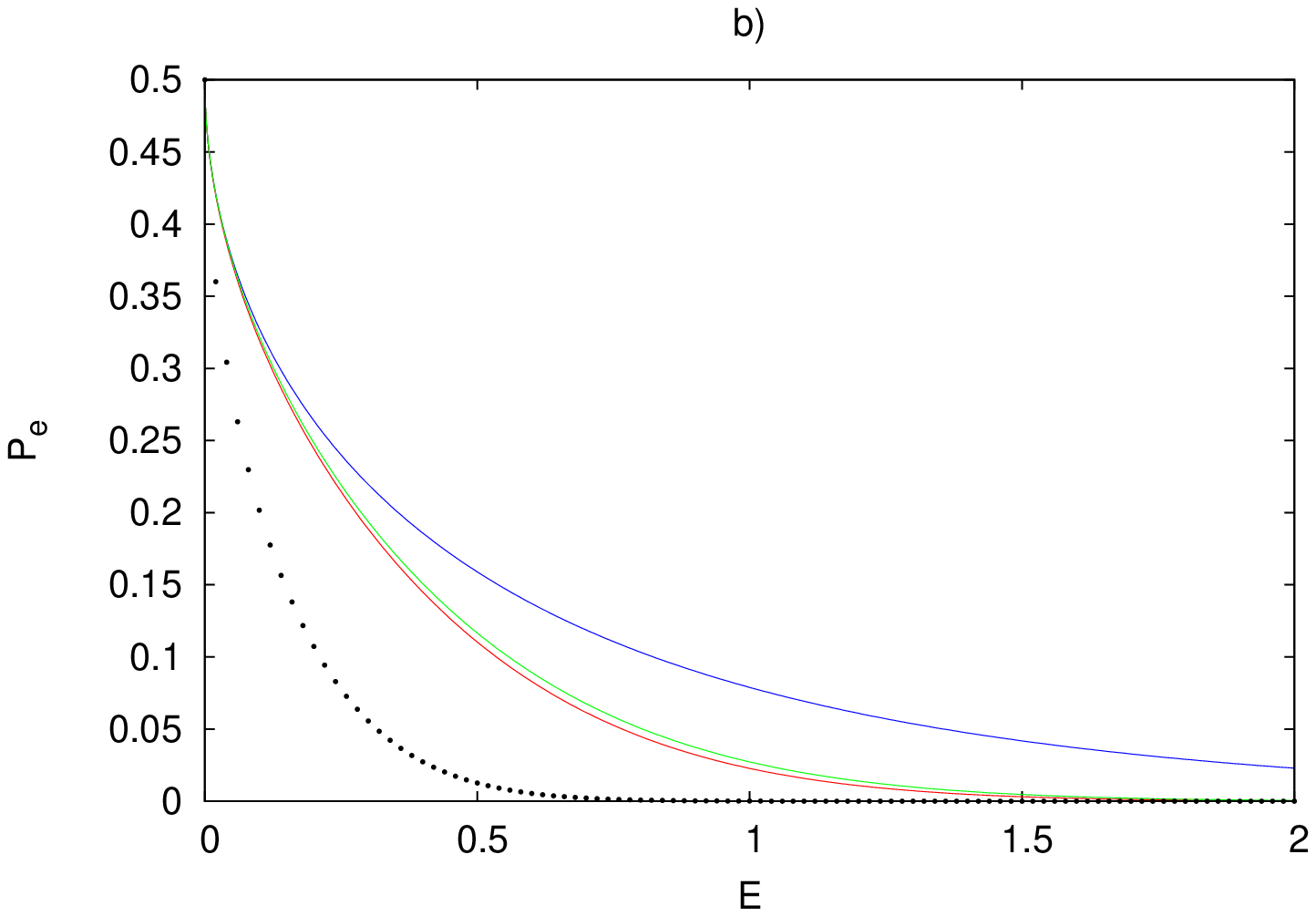}
  \caption{(Color online) Tradeoff between energy $E$ and probability
    of error $P_e$ in the Gaussian quantum reading of phase shifters
    $\map{I}$ and $\map{P}^{\delta}$, for $\delta=\pi/180$ [Fig. a)]
    and $\delta=\pi$ [Fig. b)], and for detection efficiency $\eta =
    1$ (lower red line in both figures) and $\eta = 0.9$ (middle green
    line in both figures). Notice that in Fig. b) red and green lines
    almost coincide. The Standard Quantum Limit is plotted for
    comparison (upper blue line in both figures). For the case $\delta
    = \pi$ and $\eta = 1$, the optimal tradeoff for hybrid quantum
    reading with Gaussian probe and arbitrary measurement (see
    Ref.~\cite{NYGSP12}) given by $P_e = \frac12 (1 -
    \sqrt{1-e^{-4E(E+1)}})$ is depicted (black dots) in Fig. b). The
    Figure suggests that the performance of the optimal Gaussian
    strategy is close to that of the hybrid strategy in the relevant
    regime of low probability of error.}
  \label{fig:popt}
\end{figure}
As the Figure clearly shows, the optimal Gaussian strategy for quantum
reading largely outperforms the Standard Quantum Limit, allowing to
dramatically reduce the probability of error $P_e$ for fixed energy
$E$, even in the presence of a realistic amount of loss and with
realistic limitations in the squeezing parameter. For example, when
$\delta = \pi$, setting $E = 4$ and $\eta = 0.9$ the optimal Gaussian
strategy leads to a probability of error as low as $P_e \sim 6.5 \cdot
10^{-9}$, while the Standard Quantum Limit gives $P_e \sim 2.6 \cdot
10^{-3}$.

\section{Conclusion}
\label{sect:concl}

In this work we addressed the problem of quantum reading with Gaussian
states and homodyne measurements. We showed that when no information
is encoded in the amplitude of the signal, quantum reading in the
presence of loss can be recasted to the discrimination of unitary
devices with low energy and high accuracy. We provided the optimal
Gaussian strategy for quantum reading with PSK encoding for any value
of the phase difference between the two signals encoding logical $0$
and $1$ and we showed that it dramatically outperforms the Standard
Quantum Limit even in the presence of loss and under realistic
assumptions on the practically feasible squeezing parameters. The
optimal Gaussian strategy, consisting in probing a phase shifter with
a properly tuned squeezed coherent state and performing an homodyne
measurement on the output state, is suitable for experimental
implementation with current quantum optical technology and represents
a proof of principle for an highly efficient and reliable
quantum-enhanced optical reader.

A natural generalization of the problem is to allow for different
detectors - e.g. not only homodyne, but also heterodyne or double
homodyne - and to investigate whether they can further improve the
performance of Gaussian quantum reading.

Moreover, in Ref.~\cite{BDD11} it was proven that the Standard Quantum
Limit for quantum reading is attainable without an ancillary space -
basically, because linear optical devices can not create entanglement
when their inputs are coherent states. Then, for the purpose of this
work - namely, proving that Gaussian quantum reading can outperform
the Standard Quantum Limit - it was sufficient to consider the setup
given by Eq.~\eqref{eq:setup6}.

Nevertheless, since linear optical devices can create entanglement
when their inputs are squeezed coherent states~\cite{PS94}, a further
generalization of the problem is that of entangled assisted Gaussian
quantum reading, where the setup in Eq.~\eqref{eq:setup6} is replaced
by
\begin{align*}
  \begin{aligned}
    \Qcircuit @C=\qcwidth @R=\qcheight {
      \prepareC{\ket{\alpha_0,\xi_0}} & \multigate{1}{\map{V}} &
      \ustick{\hilb{H}} \qw & \gate{\map{P}^{\delta_i}} &
      \multigate{1}{\map{W}} & \multigate{1}{\map{E}_{\bar\eta}} &
      \measureD{\Pi_0} & \cmultigate{1}{\map{J}} & \cw
      \\ \prepareC{\ket{\alpha_1,\xi_1}} & \ghost{\map{V}} &
      \ustick{\hilb{K}} \qw & \qw & \ghost{\map{W}} &
      \ghost{\map{E}_{\bar\eta}} & \measureD{\Pi_1} & \cghost{\map{J}}
      \gategroup{1}{1}{2}{2}{2mm}{--} \gategroup{1}{5}{2}{8}{2mm}{--}}
  \end{aligned},
\end{align*}
and the optimization is performed over squeezed coherent states
$\ket{\alpha_i,\xi_i}$, entangling devices $\map{V}$ and $\map{W}$,
homodyne measurements $\Pi_i$, and classical postprocessing $\map{J}$.
To understand whether entangled assisted Gaussian quantum reading can
further improve the performance of Gaussian quantum reading, one could
make use of the numerical techniques developed in Ref.~\cite{BDD11}.

\section*{Acknowledgments}

The author is grateful to Alessandro Bisio, Francesco Buscemi, Giacomo
Mauro D'Ariano, Lorenzo Maccone, Masanao Ozawa, and Massimiliano
F. Sacchi for very useful discussions, comments, and suggestions. This
work was supported by JSPS (Japan Society for the Promotion of
Science) Grant-in-Aid for JSPS Fellows No. 24-0219.


\begin{thebibliography}{}
\bibitem{note:povm} Without loss of generality the POVM for quantum
  reading among a set of $d$ devices can be taken discrete and, in
  particular, with $d$ outcomes, as the final outcome of the protocol
  is the best guess of the unknown device.
\bibitem{note:decomposition} More generally, any linear and passive
  unitary device can be decomposed in terms of beamsplitters and phase
  shifters~\cite{RZ94,SHH10,SHH11}.
\bibitem{note:psscs} The equivalence $P^\delta \ket{\alpha,\xi} =
  \ket{e^{i\delta}\alpha,e^{i2\delta}\xi}$ follows immediately from
  the fact that Fock states are eigenstates of phase shifter
  $P^\delta$, namely $P^\delta \ket{n} = e^{i\delta n} \ket{n}$, and
  from the expansion of $\ket{\alpha,\xi}$ on the Fock basis
  \cite{GA90}.
\end{thebibliography}
\end{document}